\documentclass[10pt]{revtex4-1}
\usepackage{graphicx}
\usepackage{amsmath}
\usepackage{amsfonts,amssymb}
\usepackage[utf8]{inputenc}
\newcommand{\newc}{\newcommand}
\newc{\beq}{\begin{equation}}
\newc{\eeq}{\end{equation}}
\newc{\bea}{\begin{array}}
\newc{\eea}{\end{array}}
\newcommand{\ben}{\begin{eqnarray}}
\newcommand{\een}{\end{eqnarray}}
\newc{\ra}{\rightarrow}
\newc{\bfx}{{\bf x}}
\newc{\bfV}{{\bf V}}
\newc{\cO}{{\cal O}}
\newc{\bfv}{{\bf v}}
\newc{\bfu}{{\bf u}}
\newc{\bfp}{{\bf p}}
\newc{\ve}{{\varepsilon}}
\newc{\Psibar}{\overline\Psi}
\newc{\w}{{\bf w}}
\newc{\E}{{\mathbf{E}}}
\newc{\EE}{{\mathcal E}}
\newc{\bfn}{{\mathbf\nabla}}
\newc{\la}{{\cal L}}
\newc{\tla}{{\tilde{\cal L}}}
\newc{\bp}{{\bf p}}
\newc{\ho}{\hookrightarrow }
\newc{\bP}{{\bf P}}
\newc{\pd}{{\partial}}
\newc{\piv}{{\partial_4}}
\newc{\pv}{{\partial_5}}
\newc{\bJ}{{\bf J}}
\newc{\bze}{{\mathbf 0}}
\newc{\bK}{{\bf K}}
\newc{\tphi}{{\tilde\phi}}
\newc{\tF}{{\tilde F}}
\newc{\tD}{{\tilde D}}
\newc{\tJ}{{\tilde J}}
\newc{\tj}{{\tilde j}}
\newc{\bD}{{\bf D}}
\newc{\tvphi}{{\tilde\varphi}}
\newc{\trho}{{\tilde\rho}}
\newc{\ttheta}{{\tilde\theta}}
\newc{\tpsi}{{\tilde\psi}}
\newc{\tu}{{\tilde u}}
\newc{\cD}{{\cal D}}
\newc{\tPhi}{{\tilde\Phi}}
\newc{\tPsi}{{\tilde\Psi}}
\newc{\tA}{{\tilde A}}
\newc{\talpha}{{\tilde\alpha}}
\newc{\tbeta}{{\tilde\beta}}
\newc{\bA}{{\mathbf A}}
\newc{\bB}{{\bf B}}
\newc{\br}{{\bf r}}
\newc{\sig}{{\mathbf\sigma}}
\newc{\eg}{{\rm e.g.\ }}
\newc{\ie}{{\rm i.e.\ }}
\newcommand{\pslash}{\not{\hbox{\kern-2.3pt $p$}}}
\newcommand{\pdslash}{\not{\hbox{\kern-2pt $\partial$}}}
\newtheorem{theorem}{Theorem}[section]
\newtheorem{lemma}{Lemma}[section]

\newtheorem{definition}{Definition}[section]
\newenvironment{proof}[1][Proof]{\noindent\textbf{#1.} }{\ \rule{0.5em}{0.5em}}

\begin{document}

\title[]{Quasitriangular Hopf Algebras, Braid Groups and Quantum Entanglement}%

\author{Eric Pinto}%
%\email{erpinto@ufba.br}
\affiliation{Instituto de  F\'isica, Universidade Federal da Bahia,
  Campus Ondina, 40210-340, Salvador, Bahia, Brazil.}

\author{Marco A. S. Trindade}%
\affiliation{Departamento de Ci\^{e}ncias Exatas e da Terra, Universidade do Estado da Bahia, Rodovia Alagoinhas/Salvador, BR 110, Km 03, 48040-210 - Alagoinhas, Bahia, Brazil.}

\author{J. D. M. Vianna}%
\affiliation{Instituto de F\'isica, Universidade Federal da Bahia, Campus Ondina, 40210-340, Salvador, Bahia, Brazil.}
\affiliation{Instituto de F\'isica, Universidade de Bras\'ilia, 70910-900, Bras\'ilia, DF, Brazil.}

\keywords{Quantum groups; R-matrices; entanglement; cyclic groups.}

\begin{abstract}
The aim of the paper is to provide a method to obtain representations of the braid group through a set of quasitriangular Hopf algebras. In particular, these algebras may be derived from group algebras of cyclic groups with additional algebraic structures. In this context, by using the flip operator, it is possible to construct R-matrices that can be regarded as quantum logic gates capable of preserving quantum entanglement.
\end{abstract}

\maketitle

\section{Introduction}
The discovery of quantum entanglement has its origins in the seminal article by Einstein, Podolsky and Rosen in 1935 \cite{Epr}. In this work was proposed a thought experiment that attempted to show that quantum mechanical theory was incomplete. Currently, the quantum entanglement plays key role in the quantum information and quantum computation theory \cite{Nielsen} and has been widely exploited in quantum teleportation \cite{tel}, quantum algorithms \cite{alg} and quantum cryptography \cite{Bennet1,Ekert}. An interesting proposal for quantum computing is the topological quantum computation that employs two-dimensional (2D) quasiparticles called anyons \cite{Fred}, whose world lines cross over one another to form braids in a three-dimensional (3D) spacetime. These braids form the logic gates that make up the quantum computer. One advantage of this proposal is the fact that it allows a fault-tolerant computing. Small perturbations can cause decoherence and introduce errors in computing, however these small perturbations do not change the topological properties of braids. Experimental evidences of non-Abelian anyons appear in quantum Hall systems in 2D electron gases subject to high magnetic fields \cite{Fred}.

From a mathematical perspective, the anyons are described by representations of the braid groups, an algebraic structure that was explicitly introduced by Artin in 1925 \cite{Artin}. In algebraic topology and knot theory, they can be recognized as the fundamental group of a configuration space, by using the homotopy concept \cite{Kassel1}. A bridge between knot theory and quantum information can be found in the references\cite{Kauffman1,Kauffman2,Kauffman3,Kauffman4,Sun,Zhang1}. In particular, in reference \cite{Zhang2}, the quantum teleportation has been described by the group of braids and Temperley-Lieb algebra, providing diagrammatic representations for teleportation. In this line, no explicit link is established with the anyons, with an general algebraic-topological perspective. An approach for non-Abelian anyons through quasitriangular Hopf algebras \cite{Drinfeld} was performed by Kitaev \cite{Kitaev1,Kitaev2}. The quasitriangular structure \cite{Kassel2,Majid1} provides a unified description of the braiding properties, by using the Yang-Baxter equations. The final piece is the universal R-matrix that can be used to define representations of the braid group on fusion spaces, also called topological Hilbert spaces.

Hopf algebras \cite{Hopf,Abe,Pressley,Majid2} appear naturally in algebraic topology, where  are related to the H-space concept. Its origin is in the  axiomatizations of the works of Hopf on topological properties of Lie groups. The notion of quasitriangular Hopf algebras, or quantum groups, in its turn, is due to Drinfeld  \cite{Drinfeld} as an abstraction of structures implicit in the studies of Sklyanin\cite{sklyanin1,sklyanin2}, Jimbo\cite{jimbo} and others\cite{Majid2}. There are many applications of these structures in physics, especially related to quantum gravity \cite{Pressley,Majid2}. A relationship between quantum groups and quantum entanglement can be found in references \cite{trindade,Korbicz}. Trindade and Vianna \cite{trindade} performed a possible connection between quantum groups, non-extensive statistical mechanics and entanglement through the entropic parameter q. Korbicz \textit{et al} \cite{Korbicz} addressed the problem of separability in terms of compact quantum groups, resulting in an analog of positivity of partial transpose criterion in quantum information theory. In reference \cite{Kassel2} was shown as quasitriangular Hopf algebras can generate R-matrices. This result is particularly interesting because it allows  to obtain representations of braid groups, since we have a quasitriangular structure.

In this work, we developed a general method to obtain representations of the braid groups from a set of quasitriangular Hopf algebras. We applied these results to Hopf algebras derived of cyclic group. In particular, we investigated our general results for the $CZ_{/2}$ group and obtained a quantum gate leading entangled states in themselves. We performed a comparative analysis with other work, enphasizing the differences, advantages and necessity of symmetry considerations.

The paper is organized as follows: Section 2 presents some basic concepts about quasitriangular Hopf algebras and braid groups. Section 3 contains a general method for obtaining the representations of braid groups. In Sec. 4, we derive a quantum logic gate that turns Bell states into Bell states. Section 5 is devoted to concluding remarks and outlooks.

\section{Basic Concepts}

\label{sec:construction}
In this section, we are going to review some basic concepts \cite{Majid1} that we shall need later.

\begin{definition}
Let $(H, \mu, \eta, \Delta, \varepsilon)$ be a bialgebra. We call it quasi-cocomutative if there exists an invertible element
$R$ of the algebra ${H \otimes H}$ such that for all $x$ $\in$ $H$ we have

\begin{equation}
\Delta^{op}(x)=R \Delta(x) R^{-1}
\end{equation}
\end{definition}

Here $\Delta^{op}=\tau_{H,H}\circ \Delta$ denotes the opposite coproduct on $H$, $\mu$ and $\eta$ are linear maps that express the multiplication and unit, respectively; $\Delta$ is a product and $\varepsilon$ is counity. An element $R$ satisfying this condition is called a universal R-matrix.

\begin{definition}
A quasi-cocomutative Hopf algebra $(H, \mu, \eta, \Delta, \varepsilon, S, S^{-1}, R)$ is quasitriangular if the universal
R-matrix R satisfies the two relations
\begin{equation}
(\Delta \otimes id_{H})(R)=R_{13}R_{23}
\end{equation}
and
\begin{equation}
(id_{H} \otimes \Delta)(R)=R_{13}R_{12}
\end{equation}
\end{definition}
by using Sweedler's notation \cite{Majid2} for $R_{ij}$.\\

It is possible to shown that universal R-matrix R satisfies the equation
\begin{equation}
 R_{12}R_{13}R_{23}=R_{23}R_{13}R_{12}
\end{equation}

denoted by algebraic Yang-Baxter equation. \\

\begin{definition}
Let V be a vector space over a field k. A linear automorphism $R'$ of $V \otimes V$ is said to be an R-matrix of it is a solution
of the Yang-Baxter equation
\begin{equation}
(R' \otimes id_{V})(id_{V} \otimes R')(R' \otimes id_{V})=(id_{V} \otimes R')(R' \otimes id_{V})(id_{V} \otimes R')
\end{equation}
\end{definition}
that holds in the automorphism group of $V \otimes V$.

The later relation has fundamental importance because by identification
\begin{equation}
R'_{i}=\mathbb{I}^{\otimes (i-1)} \otimes R' \otimes \mathbb{I}^{\otimes (N-i)}
\end{equation}
where $\mathbb{I}$ is on identity operator and allows to build representations of braid groups with N strands that have to satisfy
\begin{eqnarray}
R'_{i}R'_{j}=R'_{j}R'_{i},\ \ \ |i-j|\geq2\
\end{eqnarray}
and
\begin{eqnarray}
R'_{i}R'_{i+1}R'_{i}=R'_{i+1}R'_{i}R'_{i+1},\ \ \ i=1,\ldots,N-2\
\end{eqnarray}

\section{Results}

We now perform generalization about some results in Ref.\cite{Kassel2} and we explore these expressions in the context of cyclic groups.
\begin{lemma}
Let $(H_{1}, \mu_{1}, \eta_{1}, \Delta_{1}, \varepsilon_{1}, S_{1}, S^{-1}_{1}, R_{1})$,\ldots, $(H_{n},\mu_{n}, \eta_{n}, \Delta_{n},\varepsilon_{n},S_{n},S^{-1}_{n},\\ R_{n})$ quasitriangular Hopf algebras. Hence there is an invertible element R such that for all $x \in H_{1} \otimes \cdots \otimes H_{n}$, we have
\begin{equation}
\Delta^{op}(x)R=R \Delta(x)
\end{equation}

with $R_{1}=\sum_{i_{1}}s_{i_{1}} \otimes t_{i_{1}}$,\ldots, $R_{n}=\sum_{i_{n}}s_{i_{n}} \otimes t_{i_{n}}$, and
\begin{equation}
R=\sum_{i_{n},\ldots, i_{n}}s_{i_{1}}\otimes \cdots \otimes s_{i_{n}} \otimes t_{i_{1}} \otimes \cdots \otimes t_{i_{n}}
\end{equation}

Moreover the following relations
\begin{equation} \label{eq:relation2}
(\Delta \otimes id_{H_{1} \otimes \cdots \otimes H_{n}})(R)=R_{13}R_{23}
\end{equation}

\begin{equation} \label{eq:relation3}
(id_{H_{1} \otimes \cdots \otimes H_{n}} \otimes \Delta)(R)=R_{13}R_{12}
\end{equation}

\begin{equation}
R_{12}R_{13}R_{23}=R_{23}R_{13}R_{12}
\end{equation}

with $R_{12}=\sum_{i_{1} \ldots i_{n}}s_{i_{1}} \otimes \ldots s_{i_{n}} \otimes t_{i_{1}} \otimes \cdots \otimes t_{i_{n}} \otimes 1 \otimes \cdots \otimes 1$, $R_{13}=\sum_{i_{1} \ldots i_{n}}s_{i_{1}} \otimes \cdots s_{i_{n}} \otimes 1 \otimes \cdots \otimes 1 \otimes t_{i_{1}} \otimes \cdots \otimes t_{i_{n}}$ and $R_{23}=\sum_{i_{1} \ldots i_{n}} 1 \otimes \cdots \otimes 1 \otimes s_{i_{1}} \otimes \cdots \otimes s_{i_{n}} \otimes t_{i_{1}} \otimes \cdots \otimes t_{i_{n}}$ are satisfied.
\end{lemma}

\begin{proof}
Let the coproduct:
\begin{eqnarray}
\Delta(x)=\sum_{(x)}x' \otimes x'' \nonumber
\end{eqnarray}
in Sweedler's notation
\begin{eqnarray}
\Delta(x_{1} \otimes \cdots \otimes x_{n})&=&\sum_{(x_{1} \otimes \cdots \otimes x_{n})} (x_{1} \otimes x_{2} \otimes \cdots \otimes x_{n})' \otimes (x_{1} \otimes x_{2} \otimes \cdots \otimes x_{n})'' \nonumber \\
&=&\sum_{(x_{1} \otimes \cdots \otimes x_{n})} x'_{1} \otimes \cdots \otimes x'_{n} \otimes x''_{1} \otimes \cdots \otimes x''_{n} \nonumber
\end{eqnarray}
Then, in a general case, we have
\begin{eqnarray*}
\Delta^{op}(x_{1} \otimes \cdots \otimes x_{n})=\sum_{(x_{1} \otimes \cdots \otimes x_{n})} x''_{1} \otimes \cdots \otimes x''_{n} \otimes x'_{1} \otimes \cdots \otimes x'_{n}
\end{eqnarray*}
Consequently
\begin{eqnarray*}
\Delta^{op}(x_{1} \otimes \cdots \otimes x_{n})R&=&\sum_{(x_{1} \otimes \cdots \otimes x_{n})} (x''_{1} \otimes \cdots \otimes x''_{n} \otimes x'_{1} \otimes \cdots \otimes x'_{n})\sum_{i_{1} ... i_{n}} s_{i_{1}} \otimes \cdots \otimes s_{i_{n}} \otimes t_{i_{1}}\otimes \cdots \otimes t_{i_{n}} \nonumber \\
&=&\sum_{(x_{1},\ldots,x_{n};i_{1},\ldots,i_{n})}x''_{1}s_{i_{1}} \otimes x''_{2}s_{i_{2}} \otimes \cdots \otimes x''_{n}s_{i_{n}} \otimes x'_{1}t_{i_{1}} \otimes \cdots \otimes x'_{n}t_{i_{n}} \nonumber \\
&=&\left(\sum_{(x_{1},\ldots,x_{n};i_{1},\ldots,i_{n})}x''_{1}s_{i_{1}} \otimes 1 \otimes \cdots \otimes x'_{1}t_{i_{1}} \otimes \cdots \otimes 1\right) \ldots \nonumber \\
&& \left(\sum_{(x_{1},\ldots,x_{n};i_{1},\ldots,i_{n})}1 \otimes \cdots \otimes x''_{i_{n}}s_{i_{n}}
\otimes 1 \otimes \cdots \otimes x'_{n}t_{i_{n}}\right) \nonumber \\
&=&\left(\sum_{(x_{1},\ldots,x_{n};i_{1},\ldots,i_{n})}s_{i_{1}}x'_{1} \otimes 1 \otimes \cdots \otimes t_{i_{1}x''_{1}}\otimes \cdots \otimes 1\right) \ldots \nonumber \\
&&\left(\sum_{(x_{1},\ldots,x_{n};i_{1},\ldots,i_{n})}1 \otimes \cdots \otimes s_{i_{n}}x'_{i_{n}} \otimes 1 \otimes \cdots \otimes t_{i_{n}}x''_{n}\right) \nonumber \\
&=&\left(\sum_{(x_{1},\ldots,x_{n};i_{1},\ldots,i_{n})}s_{i_{1}}x'_{1} \otimes s_{i_{2}}x'_{2} \otimes \cdots \otimes s_{i_{n}}x'_{n} \otimes t_{i_{1}}x''_{1} \otimes \cdots \otimes t_{i_{n}}x''_{n}\right) \nonumber \\
&=&R\Delta(x) \nonumber
\end{eqnarray*}
For the relations (\ref{eq:relation2}) and (\ref{eq:relation3})
\begin{eqnarray}
(\Delta \otimes id_{H})\left(\sum_{i_{1} \ldots i_{n}} s_{i_{1}} \otimes \cdots \otimes s_{i_{n}} \otimes t_{i_{1}} \otimes \cdots \otimes t_{i_{n}}\right)&=&\sum_{i_{1} \ldots i_{n}}\Delta(s_{i_{1}} \otimes \cdots \otimes s_{i_{n}})\otimes  id_{H}(t_{i_{1}} \otimes \cdots \otimes t_{i_{n}}) \nonumber \\
&=&\sum_{i_{1} \ldots i_{n}}s'_{i_{1}} \otimes s'_{i_{2}} \otimes \cdots s'_{i_{n}}\otimes s''_{i_{1}} \otimes \cdots \otimes s''_{i_{n}} \nonumber \\
&&\otimes t_{i_{1}} \otimes \cdots \otimes t_{i_{n}} \nonumber \\
&=&(\sum_{i_{1}s_{1}} s'_{i_{1}} \otimes 1 \otimes \cdots \otimes s''_{i_{1}} \otimes \cdots \otimes t'_{i_{1}} \nonumber \\
&&\otimes \cdots \otimes 1) \ldots (\sum_{i_{n}s_{n}} 1 \otimes \cdots \otimes s'_{i_{n}} \otimes \cdots \otimes s''_{i_{n}} \nonumber \\
&&\otimes \cdots \otimes t_{i_{n}}) \nonumber \\
&=&(\sum_{i_{1}j_{1}} s_{i_{1}} \otimes 1 \otimes \cdots \otimes s_{j_{1}} \otimes 1 \otimes \cdots \otimes t_{i_{1}}t_{j_{1}} \nonumber \\
&&\otimes \cdots \otimes 1)...(\sum_{i_{n}j_{n}} 1 \otimes  \cdots \otimes s_{i_{n}} \otimes \cdots \otimes s_{j_{n}} \nonumber \\
&&\otimes 1 \otimes \cdots \otimes t_{i_{n}}t_{j_{n}}) \nonumber \\
&=& R_{13}R_{23} \nonumber
\end{eqnarray}
Similarly
\begin{eqnarray}
(id_{H_{1} \otimes \cdots \otimes H_{n}} \otimes \Delta)= R_{13}R_{12} \nonumber
\end{eqnarray}

For the last expression, we have
\begin{eqnarray}
R_{12}R_{13}R_{23}&=&\sum_{i_{1} \ldots i_{n}, j_{1} \ldots j_{n}, k_{1} \ldots k_{n}} s_{k_{1}}s_{j_{1}} \otimes \cdots \otimes s_{k_{n}}s_{j_{n}} \otimes t_{k_{1}}s_{i_{1}} \otimes \cdots \otimes t_{k_{n}}s_{i_{n}} \otimes t_{j_{1}}t_{i_{1}} \otimes \cdots \otimes t_{j_{n}}t_{i_{n}} \nonumber \\
 &=&(\sum_{i_{1}, j_{1}, k_{1}}s_{k_{1}}s_{j_{1}} \otimes \cdots \otimes 1 \otimes t_{k_{1}}s_{i_{1}} \otimes \cdots \otimes 1 \otimes t_{j_{1}}t_{i_{1}} \otimes \cdots \otimes 1)...(\sum_{i_{n}, j_{n}, k_{n}} 1 \otimes \cdots \nonumber \\
&& \otimes s_{k_{n}}s_{j_{n}} \otimes \cdots \otimes 1 \otimes t_{k_{n}}s_{i_{n}}\otimes \cdots \otimes t_{j_{n}}t_{i_{n}}) \nonumber \\
&=&\sum_{i_{1} \ldots i_{n}}(s_{j_{1}}s_{i_{1}} \otimes \cdots \otimes 1 \otimes s_{k_{1}}t_{k_{1}} \otimes \cdots \otimes t_{k_{1}}t_{j_{1}} \otimes \cdots \otimes 1)...(\sum_{i_{n}, j_{n}, k_{n}} 1 \otimes \cdots \otimes s_{j_{n}}s_{i_{n}} \otimes \cdots \nonumber \\
&&\otimes 1 \otimes s_{k_{n}}t_{i_{n}} \otimes \cdots \otimes t_{k_{n}}t_{j_{n}}) \nonumber \\
&=&\sum_{i_{1} \ldots i_{n}, j_{1} \ldots j_{n}, k_{1} \ldots k_{n}}s_{j_{1}}s_{i_{1}} \otimes \cdots \otimes s_{j_{n}}s_{i_{n}} \otimes \cdots \otimes s_{k_{n}}t_{i_{n}} \otimes t_{k_{1}}t_{j_{1}} \otimes \cdots \otimes t_{k_{n}}t_{j_{n}} \nonumber \\
&=& R_{23}R_{13}R_{12} \nonumber
\end{eqnarray}
\end{proof}
\ \

Let now $V_{1}, \ldots, V_{n}$ and $W_{1}, \ldots, W_{n}$ $H$-modules. We can to build a isomorphism $C^{R}_{V_{1} \ldots V_{n},W_{1} \ldots W_{n}}$ of $H_{1}, \ldots, H_{n}$ modules between $V_{1} \otimes \cdots \otimes V_{n} \otimes W_{1} \otimes \cdots \otimes W_{n}$ and $W_{1} \otimes \cdots \otimes W_{n} \otimes V_{1} \otimes \cdots \otimes V_{n}$, defined by
\begin{eqnarray}
C^{R}_{V_{1} \ldots V_{n},W_{1} \ldots W_{n}}(v_{1} \otimes \cdots \otimes v_{n} \otimes w_{1} \otimes \cdots \otimes w_{n})&=&\tau_{V_{1} \ldots V_{n},W_{1} \ldots W_{n}}(R\rhd(v_{1} \otimes \cdots \otimes v_{n} \otimes w_{1} \otimes \cdots \otimes w_{n})) \nonumber \\
&=&\sum_{i_{1} \ldots i_{n}}t_{i_{1}} \rhd w_{1} \otimes \cdots \otimes  t_{i_{n}} \rhd w_{n} \otimes  s_{i_{1}} \rhd v_{1} \otimes \cdots \otimes  s_{i_{n}} \rhd v_{n} \nonumber
\end{eqnarray}
where $\rhd$ denote the action of $H$ on $U, V$ and $W$.
\\
\begin{theorem}
For any triple $(U_{1} \otimes \cdots \otimes U_{n}, V_{1} \otimes \cdots \otimes V_{n}, W_{1} \otimes \cdots \otimes W_{n})$ of $H_{1} \otimes \cdots \otimes H_{n}$-module we have
\begin{itemize}
\item[$a)$]the map $C_{V_{1} \ldots V_{n},W_{1} \ldots W{n}}^{R}$ is an isomorphism of $H_{1} \otimes \cdots \otimes H_{n}$-module.\\
\item[$b)$]$\left(C_{V_{1}\ldots V_{n},W_{1} \ldots W_{n}}^{R} \otimes id_{U_{1} \ldots U_{n}} \right)\left(id_{V_{1} \ldots V_{n}} \otimes C_{U_{1}\ldots U_{n},W_{1} \ldots W_{n}}^{R} \right)\left(C_{U_{1}\ldots U_{n},V_{1} \ldots V_{n}}^{R} \otimes id_{W_{1} \ldots W_{n}} \right)= \\\\
    \left(id_{W_{1} \ldots W_{n}} \otimes C_{U_{1}\ldots U_{n},V_{1} \ldots V_{n}}^{R} \right)\left(C_{U_{1}\ldots U_{n},W_{1} \ldots W_{n}}^{R} \otimes id_{V_{1} \ldots V_{n}} \right)\left(id_{U_{1} \ldots U_{n}} \otimes C_{V_{1}\ldots V_{n},W_{1} \ldots W_{n}}^{R} \right)$

\end{itemize}
\end{theorem}

\begin{proof}
\begin{itemize}
\item[$a)$] For any $x_{1} \otimes \cdots \otimes x_{n} \in$ $H_{1} \otimes \cdots \otimes H_{n}$-module by using definition \\\\
 $C_{V_{1}\ldots V_{n},W_{1} \ldots W_{n}}^{R}(x_{1} \otimes \cdots \otimes x_{n})\rhd(v_{1} \otimes \cdots \otimes v_{n} \otimes w_{1} \otimes \cdots \otimes w_{n})= \tau_{v_{1} \ldots v_{n}, w_{1} \ldots w_{n}}(R\rhd\Delta(x_{1} \otimes \cdots \otimes x_{n})\rhd(v_{1} \otimes \cdots \otimes v_{n} \otimes w_{1} \otimes \cdots \otimes w_{n}))$
\\\\
By using lemma \textbf{1} and with the notation $C_{V_{1}\ldots V_{n},W_{1} \ldots W_{n}}^{R}(x_{1} \otimes \cdots \otimes x_{n})\rhd(v_{1} \otimes \cdots \otimes v_{n} \otimes w_{1} \otimes \cdots \otimes w_{n})=\mathcal{C}$, we get

\begin{eqnarray}
\mathcal{C}&=&\tau_{v_{1} \ldots v_{n}, w_{1} \ldots w_{n}}(\Delta^{op}(x_{1} \otimes \cdots \otimes x_{n})R(v_{1} \otimes \cdots \otimes v_{n} \otimes w_{1} \otimes \cdots \otimes w_{n})) \nonumber \\
&=&\tau_{v_{1} \ldots v_{n}, w_{1} \ldots w_{n}}(\sum_{x_{1} \ldots x_{n},i_{1} \ldots i_{n}}x''_{1}s_{i_{1}}\rhd v_{1} \otimes x''_{2}s_{i_{2}}\rhd v_{2} \otimes \cdots \otimes x''_{n}s_{i_{n}}\rhd v_{n} \otimes x'_{1}t_{i_{1}}\rhd w_{1} \otimes x'_{2}t_{i_{2}}\rhd w_{2} \nonumber \\
&&\otimes \cdots \otimes x'_{n}t_{i_{n}}\rhd w_{n}) \nonumber \\
&=&\sum_{x_{1} \ldots x_{n},i_{1} \ldots i_{n}}x'_{1}t_{i_{1}}\rhd w_{1} \otimes x'_{2}t_{i_{2}}\rhd w_{2} \otimes \cdots \otimes x'_{n}t_{i_{n}}\rhd w_{n} \otimes x''_{1}s_{i_{1}}\rhd v_{1} \otimes x''_{2}s_{i_{2}}\rhd v_{2} \otimes \nonumber \\
&&\cdots \otimes x''_{n}s_{i_{n}}\rhd v_{n} \nonumber \\
&=&\Delta(x_{1} \otimes \cdots \otimes x_{n})\sum_{i_{1} \ldots i_{n}}t_{i_{1}}\rhd w_{1} \otimes \cdots \otimes t_{i_{n}}\rhd w_{n} \otimes x''_{1}s_{i_{1}}\rhd v_{1} \otimes \cdots \otimes x''_{n}s_{i_{n}}\rhd v_{n} \nonumber \\
&=&\Delta(x_{1} \otimes \cdots \otimes x_{n})\tau_{v_{1} \ldots v_{n}, w_{1} \ldots w_{n}}(R\rhd[v_{1} \otimes \cdots \otimes v_{n} \otimes w_{1} \otimes \cdots \otimes w_{n}]) \nonumber \\
&=&(x_{1} \otimes \cdots \otimes x_{n})(C_{V_{1}\ldots V_{n},W_{1} \ldots W_{n}}^{R}[v_{1} \otimes \cdots \otimes v_{n} \otimes w_{1} \otimes \cdots \otimes w_{n}]) \nonumber
\end{eqnarray}
\end{itemize}

\begin{itemize}
\item[$b)$] It is easy to verify that \\\

$\left(C_{V_{1}\ldots V_{n},W_{1} \ldots W_{n}}^{R} \otimes id_{U_{1} \ldots U_{n}} \right)\left(id_{V_{1} \ldots V_{n}} \otimes C_{U_{1}\ldots U_{n},W_{1} \ldots W_{n}}^{R} \right)\left(C_{U_{1}\ldots U_{n},V_{1} \ldots V_{n}}^{R} \otimes id_{W_{1} \ldots W_{n}} \right)=$
$\sum_{i_{1} \ldots i_{n}, j_{1} \ldots j_{n}, k_{1} \ldots k_{n}}t_{k_{1}}t_{j_{1}}\rhd w_{1} \otimes \cdots \otimes t_{k_{n}}t_{j_{n}}\rhd w_{n} \otimes s_{k_{1}}t_{i_{1}}\rhd v_{1} \cdots \otimes s_{k_{n}}t_{i_{n}}\rhd v_{n} \otimes s_{j_{1}}s_{i_{1}}\rhd u_{1} \otimes \cdots \otimes s_{j_{n}}s_{i_{n}}\rhd u_{n}$

\begin{eqnarray}
&=&\left(\sum_{i_{1}, j_{1}, k_{1}}t_{k_{1}}t_{j_{1}}\rhd w_{1} \otimes \cdots \otimes s_{k_{1}}t_{i_{1}}\rhd v_{1} \otimes \cdots \otimes s_{j_{1}}s_{i_{1}}\rhd u_{1} \otimes \cdots \otimes 1 \right)\ldots \nonumber \\
&&\left(\sum_{i_{n}, j_{n}, k_{n}}1 \otimes t_{k_{n}}t_{j_{n}}\rhd w_{n} \otimes \cdots \otimes s_{k_{n}}t_{i_{n}}\rhd v_{n} \otimes \cdots \otimes 1 \otimes s_{j_{n}}s_{i_{n}}\rhd u_{n}\right) \nonumber
\end{eqnarray}
\begin{eqnarray}
&=&\left(\sum_{i_{1}, j_{1}, k_{1}}t_{j_{1}}t_{i_{1}}\rhd w_{1} \otimes \cdots \otimes t_{k_{1}}s_{i_{1}}\rhd v_{1} \otimes \cdots \otimes s_{j_{1}}s_{i_{1}}\rhd u_{1} \otimes \cdots \otimes 1 \right)\ldots \nonumber \\
&&\left(\sum_{i_{n}, j_{n}, k_{n}}1 \otimes t_{j_{n}}t_{i_{n}}\rhd w_{n} \otimes \cdots \otimes t_{k_{n}}s_{i_{n}}\rhd v_{n} \otimes \cdots \otimes 1 \otimes s_{k_{n}}s_{j_{n}}\rhd u_{n}\right) \nonumber
\end{eqnarray}
\begin{eqnarray}
&=&\left(id_{W_{1} \ldots W_{n}} \otimes C_{U_{1}\ldots U_{n},V_{1} \ldots V_{n}}^{R}  \right)\left(C_{U_{1}\ldots U_{n},W_{1} \ldots W_{n}}^{R} \otimes id_{V_{1} \ldots V_{n}} \right)\left(id_{U_{1} \ldots U_{n}} \otimes C_{V_{1}\ldots V_{n},W_{1} \ldots W_{n}}^{R} \right) \nonumber
\end{eqnarray}
\end{itemize}
by using of lemma \textbf{1}.
\end{proof}

Note that setting $U_{1}=V_{1}=W_{1}, \ldots, U_{n}=V_{n}=W_{n}$ we conclude that $C_{V_{1}\ldots V_{n},W_{1} \ldots W_{n}}^{R}$ is a solution of the Yang-Baxter equation and therefore can be used to generate representation of braid groups.\\

Consider now $Z_{/\eta_{1}}, Z_{/\eta_{2}}, \ldots, Z_{/\eta_{n}}$ be the finite cyclic groups of order $\eta_{1}, \eta_{2}, \ldots, \eta_{n}$ and $CZ_{/\eta_{1}}, CZ_{/\eta_{2}}, \ldots, CZ_{/\eta_{n}}$ be its group algebras respectively we can to build Hopf algebras \cite{Majid2} with the quasitriangular structures
\begin{eqnarray}
R_{1}=\frac{1}{\eta_{1}}\sum_{a_{1},b_{1}=0}^{\eta_{1}-1}e^{\frac{-2\pi Ia_{1}b_{1}}{\eta_{1}}}g^{a_{1}} \otimes g^{b_{1}} \end{eqnarray}
\begin{eqnarray}
R_{2}=\frac{1}{\eta_{2}}\sum_{a_{2},b_{2}=0}^{\eta_{2}-1}e^{\frac{-2\pi Ia_{2}b_{2}}{\eta_{2}}}g^{a_{2}} \otimes g^{b_{2}}
\end{eqnarray}
\begin{eqnarray}
\vdots \nonumber
\end{eqnarray}
\begin{eqnarray}
R_{n}=\frac{1}{\eta_{n}}\sum_{a_{n},b_{n}=0}^{\eta_{n}-1}e^{\frac{-2\pi Ia_{n}b_{n}}{\eta_{n}}}g^{a_{n}} \otimes g^{b_{n}} \end{eqnarray}
\\
and has coproduct $\Delta g^{a_{1}}=g^{a_{1}} \otimes g^{a_{1}}, \Delta g^{a_{2}}=g^{a_{2}} \otimes g^{a_{2}},\ldots, \Delta g^{a_{n}}=g^{a_{n}} \otimes g^{a_{n}}$. The counit is given by $\epsilon g^{a_{1}}=\epsilon g^{a_{2}}=\cdots=\epsilon g^{a_{n}}=1$ and the antipode $Sg^{a_{1}}=(g^{a_{1}})^{-1}, Sg^{a_{2}}=(g^{a_{2}})^{-1},\ldots,Sg^{a_{n}}=(g^{a_{n}})^{-1}$.\\

According to our formulation, we have\\
\begin{equation}
R=\frac{1}{\eta_{1}\eta_{2}\ldots\eta_{n}}\sum_{a_{1}, a_{2},\ldots,a_{n};b_{1}, b_{2},\ldots,b_{n}=0}^{\eta_{1}-1,\eta_{2}-1,\ldots,\eta_{n}-1}e^{\frac{-2\pi Ia_{1}b_{1}\ldots a_{n}b_{n}}{\eta_{1}\eta_{2}\ldots\eta_{n}}}g^{a_{1}} \otimes g^{a_{2}} \otimes \cdots \otimes g^{a_{n}} \otimes g^{b_{1}} \otimes g^{b_{2}} \otimes \cdots \otimes g^{b_{n}}
\end{equation}
\\
Using the notation $C_{U_{1} \ldots U_{n}, V_{1} \ldots V_{n}}^{R}(u_{1} \otimes u_{2} \otimes \cdots \otimes u_{n} \otimes v_{1} \otimes v_{2} \otimes \cdots \otimes v_{n})=\mathcal{C}_{1}$, we can show that\\
\begin{eqnarray}
\mathcal{C}_{1}&=&\frac{1}{\eta_{1}\eta_{2}\ldots\eta_{n}}\sum_{a_{1}, a_{2},\ldots,a_{n} b_{1}, b_{2}, \ldots,b_{n}}^{\eta_{1}-1,\eta_{2}-1,\ldots,\eta_{n}-1}e^{\frac{-2\pi Ia_{1}b_{1}a_{2}b_{2}\ldots a_{n}b_{n}}{\eta_{1}\eta_{2}\ldots\eta_{n}}}g^{b_{1}}\rhd v_{1} \otimes g^{b_{2}}\rhd v_{2} \otimes \cdots \otimes g^{b_{n}}\rhd v_{n} \nonumber \\
&&\otimes g^{a_{1}}\rhd u_{1} \otimes g^{a_{2}}\rhd u_{2} \otimes \cdots \otimes g^{a_{n}}\rhd u_{n}
\end{eqnarray}
\newpage

\section{Applications}
\label{sec:ff}
In order to illustrate our formalism, we consider a simple case $CZ_{/2}$ with group $G=\{\epsilon,x\}$, where $\epsilon$ is the identity. In this case

\begin{eqnarray}
R&=&\frac{1}{2}\sum_{a,b=0}^{1}e^{-\pi Iab}g^{a} \otimes g^{b} \nonumber \\
&=& \frac{1}{2}(\epsilon \otimes \epsilon + x \otimes \epsilon + \epsilon \otimes x - x \otimes x)
\end{eqnarray}

Using the regular representation $\Gamma$ of the algebra, we have
\begin{eqnarray}
\Gamma(\epsilon \otimes \epsilon)= \left(\begin{array}{cccc}
    1 & 0 & 0 & 0 \\
    0 & 1 & 0 & 0 \\
    0 & 0 & 1 & 0 \\
    0 & 0 & 0 & 1 \\
  \end{array}
  \right);
\Gamma(x \otimes \epsilon)= \left(\begin{array}{cccc}
    0 & 0 & 1 & 0 \\
    0 & 0 & 0 & 1 \\
    1 & 0 & 0 & 0 \\
    0 & 1 & 0 & 0 \\
  \end{array}
  \right) \nonumber
 \end{eqnarray}

 \begin{eqnarray}
\Gamma(\epsilon \otimes x)= \left(\begin{array}{cccc}
    0 & 1 & 0 & 0 \\
    1 & 0 & 0 & 0 \\
    0 & 0 & 0 & 1 \\
    0 & 0 & 1 & 0 \\
  \end{array}
  \right);
\Gamma(x \otimes x)= \left(\begin{array}{cccc}
    0 & 0 & 0 & 1 \\
    0 & 0 & 1 & 0 \\
    0 & 1 & 0 & 0 \\
    1 & 0 & 0 & 0 \\
  \end{array}
  \right) \nonumber
\end{eqnarray}

resulting in
\begin{equation}
\Gamma(R)= \frac{1}{2}\left(\begin{array}{cccc}
    1 & 1 & 1 & -1 \\
    1 & 1 & -1 & 1 \\
    1 & -1 & 1 & 1 \\
    -1 & 1 & 1 & 1 \\
\end{array}
\right)
\end{equation}

We next consider $\Gamma(R)\equiv R$. The associated flip operator is given by

\begin{equation}
\tau= \left(\begin{array}{cccc}
    1 & 0 & 0 & 0 \\
    0 & 0 & 1 & 0 \\
    0 & 1 & 0 & 0 \\
    0 & 0 & 0 & 1 \\
  \end{array}
  \right)
\end{equation}

Consequently,
\begin{equation}
R'= \tau R=\frac{1}{2}\left(\begin{array}{cccc}
    1 & 1 & 1 & -1 \\
    1 & -1 & 1 & 1 \\
    1 & 1 & -1 & 1 \\
    -1 & 1 & 1 & 1 \\
  \end{array}
  \right)
\end{equation}

This matrix satisfies the braid relation
\begin{equation}
(\mathbb{I} \otimes R')(R' \otimes \mathbb{I})(\mathbb{I} \otimes R')=(R' \otimes \mathbb{I})(\mathbb{I} \otimes R')(R' \otimes \mathbb{I})
\end{equation}
where $\mathbb{I}$ is a identity matrix $2 \times 2$, and it can be visualized as quantum logic gate.

Its action on the Bell states is given by
\begin{equation}
R'\left[\frac{1}{\sqrt{2}}(\left\vert 00\right\rangle+\left\vert 11\right\rangle)\right]=\frac{1}{\sqrt{2}}(\left\vert 01\right\rangle+\left\vert 10\right\rangle)= \left\vert \Psi^{+}\right\rangle
\end{equation}

\begin{equation}
R'\left[\frac{1}{\sqrt{2}}(\left\vert 01\right\rangle+\left\vert 10\right\rangle)\right]=\frac{1}{\sqrt{2}}(\left\vert 00\right\rangle+\left\vert 11\right\rangle)=\left\vert \Phi^{+}\right\rangle
\end{equation}

\begin{equation}
R'\left[\frac{1}{\sqrt{2}}(\left\vert 00\right\rangle-\left\vert 11\right\rangle)\right]=\frac{1}{\sqrt{2}}(\left\vert 00\right\rangle-\left\vert 11\right\rangle)=\left\vert \Phi^{-}\right\rangle
\end{equation}

\begin{equation}
R'\left[\frac{1}{\sqrt{2}}(\left\vert 01\right\rangle-\left\vert 10\right\rangle)\right]=\frac{1}{\sqrt{2}}(\left\vert 10\right\rangle-\left\vert 01\right\rangle)=-\left\vert \Psi^{-}\right\rangle
\end{equation}\\

i.e. the entanglement is preserved under the action of this gate. Importantly, the symmetry groups in 3D are cyclic as abstract group. Therefore, the cyclic groups may indirectly reflect symmetries of physical systems transforming maximally entangled states in themselves. Interestingly, starting from on extremely simple case, it is possible to generate a nontrivial structure.\\

In a seminal work that established a connection between quantum entanglement and topological entanglement, Kauffman and Lomanaco Jr.\cite{Kauffman3} introduced the following matrix solution to the Yang-Baxter equation
\begin{eqnarray*}
R = \left(\begin{array}{cccc}
    a & 0 & 0 & 0 \\
    0 & 0 & d & 0 \\
    0 & c & 0 & 0  \\
    0 & 0 & 0 & b \\
\end{array}
\right)
\end{eqnarray*}
where a, b, c and d are any scalars on the unit circle in the complex plane. It was show that if $R$ is chosen so that $ab\neq cd$, then the state $R(\psi \otimes \psi)$, with $\psi=\left\vert 0\right\rangle + \left\vert 1\right\rangle$, is entangled.\\

All the 4x4 unitary matrix solutions to the braided Yang-Baxter equation was obtained by Dye \cite{Dye1}. For this dimension, the relationship between quantum entanglement and topological entanglement was analyzed and families of solutions have been classified. A solution very explored by Zhang \cite{Zhang3} is given by the matrix:

  \begin{eqnarray*}
B = \frac{1}{2}\left(\begin{array}{cccc}
    1 & 0 & 0 & 0 \\
    0 & 1 & -1 & 0 \\
    0 & 1 & 1 & 0  \\
    -1 & 0 & 0 & 1 \\
\end{array}
\right)
\end{eqnarray*}
called Bell matrix. The action of this matrix on the basis state results in Bell states:

\begin{eqnarray*}
B\left\vert 00\right\rangle & = & \frac{1}{\sqrt{2}}(\left\vert 00\right\rangle-\left\vert 11\right\rangle)= \left\vert \Phi^{-}\right\rangle  \\
B\left\vert 01\right\rangle & = & \frac{1}{\sqrt{2}}(\left\vert 01\right\rangle+\left\vert 10\right\rangle)= \left\vert \Psi^{+}\right\rangle   \\
B\left\vert 10\right\rangle & = & -\frac{1}{\sqrt{2}}(\left\vert 01\right\rangle+\left\vert 10\right\rangle)=- \left\vert \Psi^{-}\right\rangle   \\
B\left\vert 11\right\rangle & = & \frac{1}{\sqrt{2}}(\left\vert 00\right\rangle+\left\vert 11\right\rangle)=\left\vert \Phi^{+}\right\rangle   \\
\end{eqnarray*}

In this context, equations for teleportation and their diagrammatical representations were presented. Here, unlike the mentioned above approaches, our purpose is to obtain representations of braid groups in a systematic way for arbitrary dimensions, exploiting the structure of underlying symmetries. This is interesting because anyonic physics, for example, Monchon \cite{monchon} showed that for anyons obtained from a finite gauge theory, the computational power depends on the symmetry group. Besides that, for cyclic groups, according to recent investigations about anyons in integer quantum Hall magnets \cite{muniz}, the nontrivial fundamental homotopy group $\pi_{1}(O(3))=Z_{2}$ guarantees the existence of the $Z_{2}$ vortices. It is noteworthy that although our analysis have been done for cyclic groups, any group or algebra could have been used, since our method is general.

\section{Conclusions}

The main objective of this work has been to present a systematic method to derive representations of braid groups through a set of quasitriangular Hopf algebras. This approach should be related to the topological quantum computation. In reference \cite{vidal}, possible experimental implementations of lattice models based on non-Abelian discrete symmetry groups has been proposed. The group elements can be viewed as transformations between the states of the sites of an superconducting Josephson-junction array. The algebraic structure employed was the dihedral group that can be expressed in terms of cyclic groups using the semidirect product. In this paper, we generalize some results obtained in reference \cite{Kassel2} and we explore the structure of a quasitriangular Hopf algebra derived from a cyclic group. In particular, we show how to obtain a quantum logic gate of a simple abelian structure generated by $CZ_{/2}$ group algebra. This gate becomes entangled states in themselves. Furthermore, we compared our method with some related works, highlighting differences and possible advantages if symmetry considerations are required. As perspectives, it seems interesting to investigate other cyclic groups, as well as, possible associated topological structure.

\end{document}